\documentclass[12pt]{article}
\usepackage[utf8]{inputenc}
\usepackage{amssymb}
\usepackage{amsfonts}
\usepackage{amsmath}
\usepackage{amsthm}
\usepackage{mathrsfs} 
\usepackage{graphicx}
\usepackage{mathtools}
\usepackage{enumitem}
\usepackage{caption}
\numberwithin{equation}{section}
\usepackage{xcolor}

\usepackage{yfonts}
\usepackage{amssymb}
\usepackage{amsmath}
\theoremstyle{theorem}
\newtheorem{theorem}{Theorem}[section]
\newtheorem{lemma}[theorem]{Lemma}

\newtheorem{definition}[theorem]{Definition}
\newtheorem{remark}[theorem]{Remark}

\newtheorem{assumption}[theorem]{Assumption}

\newtheorem{corollary}[theorem]{Corollary}

\def\bearray{\begin{eqnarray}}
\def\earray{\end{eqnarray}}
\def\beq{\begin{equation}}
\def\eeq{\end{equation}}

\def\b0{{\bf 0}}

\newsymbol\bt 1202             

\def\bC{{\mathbb C}}           

\def\bN{{\mathbb N}}

\def\bR{{\mathbb R}}

\newsymbol\rest 1316         

\begin{document} 

\par
\bigskip
\large
\noindent
{\bf The classical limit of mean-field quantum spin systems}
\bigskip
\par
\rm
\normalsize

\noindent {\bf Christiaan J.F.  van de Ven$^{a}$$^1$}\\
\par

\noindent 
 $^a$ Department of  Mathematics, University of Trento, and INFN-TIFPA \\
 via Sommarive 14, I-38123  Povo (Trento), Italy\\
$^1$ Marie Sk\l odowska-Curie Fellow of the Istituto Nazionale di Alta Matematica. Email:
christiaan.vandeven@unitn.it\\

\par

\rm\normalsize

\date{today}

\rm\normalsize


\par
\bigskip

\noindent
\small
{\bf Abstract}
The theory of strict deformation quantization of the two sphere $S^2\subset\mathbb{R}^3$ is used to prove the existence of the classical limit of mean-field quantum spin chains, whose ensuing Hamiltonians are denoted by $H_N$ and where $N$ indicates the number of sites. Indeed, since  the fibers  $A_{1/N}=M_{N+1}(\mathbb{C})$ and $A_0=C(S^2)$ form a continuous bundle of $C^*$-algebras over the base space $I=\{0\}\cup 1/\mathbb{N}^*\subset[0,1]$, one can define a strict deformation quantization of $A_0$ where quantization is specified by certain quantization maps $Q_{1/N}: \tilde{A}_0 \rightarrow A_{1/N}$, with $\tilde{A}_0$ a dense Poisson subalgebra of $A_0$. Given now a sequence of such $H_N$, we show that under some assumptions a sequence of eigenvectors $\psi_N$ of $H_N$ has a classical limit in the sense that $\omega_0(f):=\lim_{N\to\infty}\langle\psi_N,Q_{1/N}(f)\psi_N\rangle$ exists as a state on $A_0$ given by $\omega_0(f)=\frac{1}{n}\sum_{i=1}^nf(\Omega_i)$, where $n$ is some natural number. We give an application regarding spontaneous symmetry breaking (SSB) and moreover we show that the spectrum of such a mean-field quantum spin system converges to the range of some polynomial in three real variables restricted to the sphere $S^2$ .
\normalsize
\newpage
\section{Introduction}
The concept of strict deformation quantization has been introduced by Rieffel in \cite{Rie89,Rie94} in order to provide a mathematical formalism describing the transition from a classical theory to a
quantum theory in terms of deformations of (commutative) Poisson algebras (representing the classical theory) into non-commutative $C^*$-algebras (characterizing the quantum theory).  Concretely  a strict deformation quantization of a commutative Poisson algebra $A_0$ consists of a continuous bundle $\mathcal{A}$ of $C^*$-algebras $(A_{\hbar})_{\hbar\in I}$ over an interval $I$ along with a family of quantization maps $Q_{\hbar} : \tilde{A}_0 \to A_{\hbar}$, with $\hbar\in I$ and $\tilde{A}_0\subset A_0$ a dense Poisson subalgebra on which the the quantization maps $Q_{\hbar}$ are defined (viz. Definition \ref{def:deformationq}).  This machinery allows us to prove the existence of physical phenomena often arising in the classical theory, but which do not exist in the corresponding underlying quantum theory. A typical example of such an emergent feature is spontaneous symmetry breaking (SSB), discussed in section \ref{SSB1}. In this paper we apply this formalism to quantum spins systems where quantization of the manifold $S^2$ is exploited. In particular, we study the asymptotic relation between a classical theory on the two sphere and statistical mechanics of a quantum (mean-field) spin system on a finite lattice, as the number of lattice sites increases to infinity.  
\subsection{Strict deformation quantization}
Before giving the mathematical definitions, let us first introduce a basic example concerning the quantization of the classical phase space $\mathbb{R}^{2n}$ on which classical mechanics of a particle is described. For convenience we take the simplest functional-analytic setting in which only smooth compactly supported functions $f\in C_c(\mathbb{R}^{2n})$ (with Poisson structure given by the natural symplectic form $\sum_{j=1}^ndp_j\wedge dq^j$) are quantized. In order to relate $C_c(\mathbb{R}^{2n})$ to a quantum theory described on some Hilbert space, one needs to deform $C_c(\mathbb{R}^{2n})$ into non-commutatative $C^*$-algebras
exploiting a family of quantization maps. In this setting Berezin introduced the maps \cite{Ber}
\begin{align}
&Q_{\hbar}:  C_c^{\infty}(\mathbb{R}^{2n}) \to B_0(L^2(\mathbb{R}^n));\\
&Q_{\hbar}(f)=\int_{\mathbb{R}^{2n}} \frac{d^npd^nq}{(2\pi\hbar)^n} f(p,q) |\phi_{\hbar}^{(p,q)}\rangle \langle\phi_{\hbar}^{(p,q)}|,
\end{align}
 where  $\hbar\in (0,1]$; $B_0(H)$ is the $C^*$-algebra of compact operators on the Hilbert space $H=L^2(\bR^n)$,   and for each point $(p,q)\in \bR^{2n}$ the (projection) operator $|\phi_{\hbar}^{(p,q)}\rangle \langle \phi_{\hbar}^{(p,q)}|:L^2(\bR^n)\to L^2(\bR^n)$ is induced by the normalized wavefunctions, where $x\in \bR$,
\begin{equation}
\phi_{\hbar}^{(p,q)}(x)=(\pi\hbar)^{-n/4}e^{-ipq/2\hbar}e^{-ipx/\hbar}e^{-(x-q)^2/2\hbar}\:, \quad \phi_{\hbar}^{(p,q)}\in L^2(\mathbb{R}),
\end{equation}
defining the well-known  (Schrodinger) {\em coherent states}. Inspired by Dixmier's concept of a {\em continuous bundle} \cite{Dix}, Rieffel  showed that \cite{Rie89,Rie94} 
\begin{enumerate}
\item The fibers $A_0= C_0(\bR^{2n})$ and  $A_{\hbar}= B_0(L^2(\bR^n))$, $h\in(0,1]$, can be combined into a (locally non-trivial) continuous bundle $A$ of $C^*$-algebras over $I=[0,1]$;
\item  $\tilde{A}_0=C^{\infty}_c(\bR^{2n})$ is a  dense Poisson subalgebra of $A_0$.
 \item  Each quantization map  $Q_{\hbar}:\tilde{A}_0\to A_{\hbar}$ is linear, and if we also define $Q_0: \tilde{A}_0 \hookrightarrow A_0$ as the inclusion  map, then the ensuing family $Q = (Q_\hbar)_{\hbar \in I}$
 satisfies:
  \begin{enumerate}[label=(\arabic*)]
\item  Each map $Q_{\hbar}$ is self-adjoint, i.e.\ $Q_{\hbar}(\overline{f})=Q_{\hbar}(f)^*$ (where $f^*(x)=\overline{f(x)}$). 
\item  For each $f\in \tilde{A}_0$ the following cross-section of the bundle is continuous:
\begin{align}
    &0\to f;\\
    &\hbar\to Q_{\hbar}(f) \ \ \  (\hbar\in I\backslash\{0\})).
\end{align}
\item  Each pair $f,g\in\tilde{A}_0$ satisfies the {\bf Dirac-Groenewold-Rieffel condition}:
\begin{align}
    \lim_{\hbar\to 0}\left|\left|\frac{i}{\hbar}[Q_{\hbar}(f),Q_{\hbar}(g)]-Q_{\hbar}(\{f,g\})\right|\right|_{\hbar}=0.\label{Diracgroenewold}
\end{align} 
\end{enumerate}
\end{enumerate}
These obervations led to the general concept of a {\em strict deformation of a Poisson manifold} $X$ \cite{Rie89,Lan98}, which we here state in the case of interest to us in which $X$ is compact or in which $X$ is a manifold with stratified boundary \cite{LMV,Pflaum}. In that case the space $I$ in which $\hbar$ takes values cannot be all of $[0,1]$, but should be a subspace $I\subset [0,1]$ thereof that at least contains 0 as an accumulation point. The reason for this is to ensure that the continuity properties of the bundle (viz. Remark \ref{continuitybundle} below) are met. This is assumed in what follows. Furthermore, the  Poisson bracket on $X$ is denoted, as usual, by $\{\cdot, \cdot\}: C^\infty(X)\times C^\infty(X) \to C^{\infty}(X)$ (where the smooth space $C^\infty(X)$ is suitably defined when $X$ is a more complicated object than a compact smooth manifold \cite{LMV}). 
\begin{definition}\label{def:deformationq}
Let $I\subset\mathbb{R}$ containing $0$ as accumulation point.
A {\bf strict deformation quantization} of a compact Poisson manifold $X$ on $I$ consists of
 \begin{itemize}
\item  A continuous
 bundle of unital $C^*$-algebras $(A_{\hbar})_{\hbar\in I}$ over $I$  with $A_0=C(X)$;
 \item A dense Poisson subalgebra  $\tilde{A}_0 \subseteq C^\infty(X) \subset A_0$ (on which $\{\cdot, \cdot\}$ is defined);
 \item  A  family $Q = (Q_\hbar)_{\hbar \in I}$ of  linear maps $Q_{\hbar}:\tilde{A}_0\to A_{\hbar}$ indexed by $\hbar\in I$ (called  {\bf quantization maps})  such that $Q_0$ is the inclusion map $\tilde{A}_0 \hookrightarrow A_0$, and the next
  conditions (a) - (c) hold, as well as  $Q_{\hbar}(\mathrm{1}_X)=\mathrm{1}_{A_{\hbar}}$ (the unit of $A_{\hbar}$)
 \begin{enumerate}
\item  Each map $Q_{\hbar}$ is self-adjoint, i.e.\ $Q_{\hbar}(\overline{f})=Q_{\hbar}(f)^*$ (where $f^*(x)=\overline{f(x)}$). 
\item  For each $f\in \tilde{A}_0$ the following cross-section of the bundle is continuous:
\begin{align}
    &0\to f;\\
    &\hbar\to Q_{\hbar}(f) \ \ \  (\hbar\in I\backslash\{0\})).
\end{align}
\item  Each pair $f,g\in\tilde{A}_0$ satisfies the {\bf Dirac-Groenewold-Rieffel condition}:
\begin{align}
    \lim_{\hbar\to 0}\left|\left|\frac{i}{\hbar}[Q_{\hbar}(f),Q_{\hbar}(g)]-Q_{\hbar}(\{f,g\})\right|\right|_{\hbar}=0.\label{Diracgroenewold}
\end{align} 
\end{enumerate}
\end{itemize}
\end{definition}
\begin{remark}\label{continuitybundle}
It follows from the definition of a continuous bundle of $C^*$-algebras that continuity properties like
\begin{align}
\lim_{\hbar\to 0}\|Q_{\hbar}(f)\|_\hbar=\|f\|_{\infty},
\end{align}
 and 
\begin{align}
\lim_{\hbar\to 0}\|Q_{\hbar}(f)Q_{\hbar}(g)-Q_{\hbar}(fg)\|_\hbar=0,
\end{align}
hold automatically.
\end{remark}
{\em Mean-field quantum spin systems} (see e.g. \cite{ABN,CCIL,IL,VGRL18,vandeVen} and references therein) fit into this framework. There,  the index set $I$ is given by ($0\notin\bN^*:=\{1,2,3,\ldots\}$)
\begin{equation}
I= \{1/N \:|\: N \in  \mathbb{N}^*\} \cup\{0 \}\equiv (1/\bN^*) \cup\{0 \}, \label{defI}\
\end{equation}
with the  topology inherited from $[0,1]$.  That is, we  put $\hbar=1/N$, where $N\in\mathbb{N}^*$ is interpreted as  the number of sites of the model; our interest is the limit $N\to\infty$. 
In this setting the above definition particularly applies when $X=S^2\subset\mathbb{R}^3$, i.e. the classical phase space for quantum mean-field spin systems. Indeed, the following shows the existence of a strict deformation quantization of $S^2$ with Poisson bracket given on $C^{\infty}(S^2)$ by
\begin{align}
\{f,g\}(x)=\sum_{a,b,c=1}^3\varepsilon_{ab}^cx_c\frac{\partial f(x)}{\partial{x_a}}\frac{\partial g(x)}{\partial{x_b}},
\end{align}
where $\varepsilon_{ab}^c$ is the Levi-Civita symbol.
Let us indicate the algebra of bounded operators on $\text{Sym}^N(\mathbb{C}^2)$ by $B(\text{Sym}^N(\mathbb{C}^2))$. Here $\text{Sym}^N(\mathbb{C}^2)\subset\bigotimes_{n=1}^N\mathbb{C}^2$ is the symmetric subspace of dimension $N+1$, given by the {\em symmetric} $N$-fold tensor product of $\mathbb{C}^2$ with itself. The reason we recall the space $\text{Sym}^N(\mathbb{C}^2)$ rather than the isomorphic space $\mathbb{C}^{N+1}$ is because of our application to mean field quantum spin chains, whose Hamiltonians originally defined on the $N$-fold tensor-product of $\mathbb{C}^2$ with itself, typically leave the symmetric subspace  $\text{Sym}^N(\mathbb{C}^2)$ invariant. More presisely, $\text{Sym}^N(\mathbb{C}^2)$ is given by the range of the symmetrizer operator $P_N$, defined by linear extension of the following map $P_N$ on elementary tensors $(v_1\otimes\cdot\cdot\cdot\otimes v_N)$:
\begin{align}
P_N (v_1 \otimes \cdots \otimes v_N) = \frac{1}{N!} \sum_{\sigma \in {\cal P}(N)} v_{\sigma(1)} \otimes \cdots \otimes v_{\sigma(N)}. \label{defPN}
\end{align}
It is known \cite[Theorem 8.1]{Lan17} that
 \begin{align}
 A_0&=C(S^2); \label{B02}\\
 A_{1/N}&=M_{N+1}(\mathbb{C})\simeq B(\text{Sym}^N(\mathbb{C}^2))\label{BN2},
   \end{align}
are the fibers of a continuous bundle of $C^*$-algebras over base space  $I$ defined by
\begin{align}
I=\{1/N\ |N\in\mathbb{N}^*\}\cup\{0\} = (1/\mathbb{N}^*)\cup \{0\}.
\end{align}
The continuous cross-sections are given by all sequences $(a_{1/N})_{N\in\mathbb{N}}\in\Pi_{n\in\mathbb{N}}A_{1/N}$ for which $a_0\in C(S^2)$ and $a_{1/N}\in M_{N+1}(\mathbb{C})$  and  such that the
sequence  $(a_{1/N})_{N\in\mathbb{N}}$ is asymptotically equivalent to $(Q_{1/N}(a_0))_{N\in\mathbb{N}}$, in the sense that
\begin{align}
\lim_{N\to\infty}||a_{1/N}-Q_{1/N}(a_0)||_N=0.\label{equivalencebookklaas}
\end{align}
Here, the symbol $Q_{1/N}$ denotes the quantization maps
\begin{align}
Q_{1/N}:\tilde{A}_0 \to A_{1/N},
\end{align}
where $\tilde{A}_0\subset C^{\infty}(S^2)\subset A_0$ is the dense Poisson subalgebra made of polynomials in three real variables restricted to the sphere $S^2$. The maps $Q_{1/N}$ are defined by the integral computed in weak sense
\begin{align}
Q_{1/N}(p)& :=
 \frac{N+1}{4\pi}\int_{S^2}p(\Omega)|v^{(\Omega)}\rangle\langle v^{(\Omega)}|_Nd\Omega\: \label{defquan3},
\end{align}
where $p$ denotes an arbitrary polynomial restricted to  $S^2$, $d\Omega$ indicates the unique $SO(3)$-invariant Haar measure on ${S}^2$ with $\int_{{S}^2} d\Omega = 4\pi$, and $|v^{(\Omega)}\rangle\langle v^{(\Omega)}|_N$ which are elements in $B(\text{Sym}^N(\mathbb{C}^2))$ denote the orthogonal projectors associated to the $N$ coherent spin states (see Appendix \ref{appB} for a general construction). In particular, if $1$ is the constant function $1(\Omega)=1_N, \ (\Omega\in S^2)$, and $1_N$ is the identity on $A_{1/N}=B(\text{Sym}^N(\mathbb{C}^2))$, the previous definition implies 
\begin{align}
Q_{1/N}(1) = 1_N  \label{defquan4}\:.
\end{align}
Indeed, it can be shown that the quantization maps \eqref{defquan3} satisfy the axioms of Definition \ref{def:deformationq}, which yields the existence of a strict deformation quantization of $S^2$. We remark that several equivalent definitions of these quantization maps are used in literature, e.g. \cite{Lan17,Pe72}. Moreover, since $S^2$ is a special case of a regular integral coadjoint orbit in the dual of the Lie algebra associated to $SU(2)$ which can be identified with $\mathbb{R}^3$, this theory can be generalized to arbitrary compact connected Lie groups \cite{Lan98}.  These quantization maps,  constructed from a family coherent states, also define a so-called {\bf pure state Berezin quantization} \cite{Lan98} for which \eqref{Berprop}, viz.
\begin{align}
f(\Omega)=\lim_{N\to\infty}\frac{N+1}{4\pi}\int_{S^2}d\Omega' f(\Omega')|\langle v^{(\Omega')},v^{(\Omega)}\rangle_N|^2, \ \ (\Omega\in S^2) \label{Berprop}
\end{align}
typically holds as well as surjectivity and positivity, in that $Q_{1/N}(f)\geq 0$ if $f\geq 0$ almost everywhere on $S^2$. 
Moreover, for all $N\in\mathbb{N}$ one has
\begin{align}
1=\frac{N+1}{4\pi}\int_{S^2}d\Omega'|\langle v^{(\Omega')},v^{(\Omega)}\rangle_N|^2. \label{identityunit}
\end{align}
\subsection{Mean-field theories and symbol}\label{Mean field theories and symbol}
In this section we consider {\em homogenous mean-field} theories, which are defined by a
single-site Hilbert space $\mathcal{H}_x = \mathcal{H} = \mathbb{C}^{n}$ and local Hamiltonians of the type
\begin{align}
H_{\Lambda}=|\Lambda|\tilde{h}(T_0^{(\Lambda)},T_1^{(\Lambda)},\cdot\cdot\cdot, T_{n^2-1}^{(\Lambda)}),\label{meanfield}
\end{align}
where $\tilde{h}$ is a polynomial on $M_n(\mathbb{C})$, and $\Lambda$ denotes a finite lattice on which $H_{\Lambda}$ is defined \cite{Lan17}. Here $T_0= 1_{M_n(\mathbb{C})}$, and the matrices $(T_i)_{i=1}^{n^2-1}$ in $M_n(\mathbb{C})$ form a basis of the real vector space of traceless self-adjoint $n\times n$ matrices; the latter may be identified with $i$ times the Lie algebra $\mathfrak{su(n)}$ of $SU(n)$, so that $(T_0,T_1,...,T_{n^2-1})$ is a basis of $i$ times the Lie algebra $\mathfrak{u(n)}$ of the unitary group $U(n)$ on $\mathbb{C}^n$. In those terms, we define
\begin{align}
T_i^{(\Lambda)}=\frac{1}{|\Lambda|}\sum_{x\in\Lambda}T_i(x).
\end{align}

Let us now introduce the notion of a {\em classical symbol}, i.e. a function
\begin{align}
h_N:=\sum_{k=0}^{M}N^{-k}h_k+O(N^{\infty}), \ \ \label{classicalsymbol1}
\end{align}
for some $M\in\mathbb{N}$ and where each $h_k$ is a smooth real-valued function on the manifold $M$ one considers. The first term $h_0$ is called the {\em principal symbol}. In the case for mean-field quantum theories we will see below that the principal symbol exactly plays the role of the polynomial $\tilde{h}$ defined above.  We remark that a general mean-field model is defined on a lattice with $N$ sites. However, as every spin interacts with every other spin the geometric configuration including the dimension of the lattice is irrelevant \cite{LMV,VGRL18}. We therefore restrict ourselves to mean field quantum spin chains, i.e. we consider tensor products of $M_2(\mathbb{C})$. Moreover, as already mentioned their Hamiltonians share the property that they leave the symmetric subspace $\text{Sym}^N(\mathbb{C}^2)\subset\bigotimes_{n=1}^N\mathbb{C}^2$ invariant. In what follows we consider mean-field quantum spin systems whose Hamiltonians $H_{1/N}$ are restricted to this subspace, since quantum spin systems arising in that way are tipically of the form $Q_{1/N}(h_N)$ for some ($N$-dependent) real polynomial function $h_N$ on $M=S^2$ given by \eqref{classicalsymbol1} and $Q_{ 1/N}$ given by \eqref{defquan3}  (see e.g. \cite[Theorem 3.1]{MV} for an example). In a different setting this connection was already noticed by Ettore Majorana, proposing that pure (vector) states induced by a sequence of eigenvectors of an $N$-qubit system which are permutation-invariant correspond to pure spin-$J=N/2$ quantum states which in turn are represented by $N$ points on the Bloch sphere $S^2$ \cite{Mar}. Such spin systems are widely studied in condensed matter physics, but also in mathematical physics they form an important field of research. One tries to calculate the corresponding partition function, or quantities like the free energy or the entropy of the system in question and considers their thermodynamic limit as the number of sites $N$ increas to infinity. 
\\\\
{\em Example 1.}
Consider the quantum Curie-Weiss Hamiltonian on a chain scaled by a factor $1/(N+2)$:
\begin{align}
H^{CW}_{1/N}: &  \underbrace{\mathbb{C}^2 \otimes \cdots  \otimes\mathbb{C}^2}_{N \: times}  \to 
\underbrace{\mathbb{C}^2 \otimes \cdots  \otimes\mathbb{C}^2}_{N \: times};\nonumber \\
H^{CW}_{1/N} &=\frac{1}{N+2} \left(-\frac{J}{2N} \sum_{i,j=1}^N \sigma_3(i)\sigma_3(j) -B \sum_{j=1}^N \sigma_1(j)\right).\label{CWhamscaled}
\end{align}
Here $\sigma_k(j)$ stands for $I_2 \otimes \cdots \otimes \sigma_k\otimes \cdots \otimes I_2$, where $\sigma_k$ occupies  the $j$-th slot, and 
 $J,B \in \mathbb{R}$ are given constants defining the strength of the spin-spin coupling and the (transverse) external magnetic field, respectively.
Regarding \eqref{meanfield} is it not difficult to see that
\begin{align}
\tilde{h}^{CW}(T_1,T_2,T_3)=-2(JT_3^2+BT_1), \label{tildepol}
\end{align}
where $T_i=\frac{1}{2}\sigma_i$. In order to determine the symbol we show that the Hamiltonian is a quantization of the {\em classical Curie-Weiss Hamiltonian}, i.e. a polynomial $h_0^{CW}$ on $S^2$ given (in spherical coordinates) by
\begin{align}
 h_0^{CW}(\theta,\phi)=-(\frac{J}{2}\cos{(\theta)}^2+B\sin{(\theta)}\cos{(\phi)}).
\end{align}
Hereto we recall a result originally obtained by Lieb \cite{Lieb}, namely that under the maps $Q_{1/N}$ given by \eqref{defquan3} one has a correspondence between functions $G$ (also called upper symbol) on the sphere $S^2$ and operators $A_G$ on $\mathbb{C}^{N+1}$ such that they satisfy the relation $A_G=Q_{1/N}(G)$. For some spin operators, the functions $G$ are determined (see Table $1$ below).
\begin{table}[ht]
\begin{center}
\begin{tabular}
{ |p{3cm}||p{9cm}|p{3cm}|p{3cm}|}
 \hline
 $\text{Spin Operator}$ & $G(\theta,\phi)$ \\
 \hline
 $S_z$   &  $\frac{1}{2}(N+2)\cos{(\theta)}$\\
 $S_z^2$ &  $\frac{1}{4}(N+2)(N+3)\cos{(\theta)}^2-\frac{1}{4}(N+2)$\\
 $S_x$ & $\frac{1}{2}(N+2)\sin{(\theta)}\cos{(\phi)}$ \\
 $S_x^2$ & $\frac{1}{4}(N+2)(N+3)\cos{(\phi)}\sin {(\theta)}-\frac{1}{4}(N+2)$\\
 $S_y$ & $\frac{1}{2}(N+2)\sin{(\theta)}\sin{(\phi)}$ \\
 $S_y^2$ &  $\frac{1}{4}(N+2)(N+3)\sin{(\phi)}\sin{(\theta)}-\frac{1}{4}(N+2)$ \\
 \hline
\end{tabular}
\caption{Spin operators on $\text{Sym}^N(\mathbb{C}^2)\simeq\mathbb{C}^{N+1}$ and their corresponding upper symbols $G$.}
\end{center}
\end{table}
Here, $S_{z}$ is the total spin operator in the $z$-direction: $S_z=\frac{1}{2}\sum_k\sigma_z(k)$, $S_x$ is the total spin operator in the $x$-direction, and $S_y$ is the total spin operator in the $y$-direction. Using these results, a straightforward computation shows that
\begin{align}
H^{CW}_{1/N}|_{\text{Sym}^N(\mathbb{C}^2)} = Q_{1/N}(h_0^{CW})-\frac{3J}{N}Q_{1/N}(z^2)+ \frac{1}{N}Q_{1/N}(1).
\end{align}
We write $h_N:=h_0^{CW}+N^{-1}(-3Jz^2+1)$, so that by linearity of $Q_{1/N}$ one has $Q_{1/N}(h_N)=H_{1/N}^{CW}|_{\text{Sym}^N(\mathbb{C}^2)}$. The function $h_N:S^2\to\mathbb{R}$ is the classical symbol associated to the quantum spin operator $H_{1/N}^{CW}|_{\text{Sym}^N(\mathbb{C}^2)} $ with principal symbol $h_0^{CW}$ which indeed has the same form as $\tilde{h}^{CW}$ in \eqref{tildepol}.
\\\\
{\em Example 2.}
Let us consider the Lipkin-Meshkov-Glick (LMG) model. This model was first proposed to describe phase transitions in atomic nuclei \cite{Lip,Dus}, and more recently it was found that the LMG model is relevant to many other quantum systems, such as cavity QED \cite{Mor}. The (scaled) Hamiltonian of a general  LMG model is given by
\begin{align}
    H_{1/N}^{LMG}=-\frac{\lambda}{N(N+2)}(S_{x}^2+\gamma S_{y}^2)-\frac{B}{N+2}S_z,\label{eq:LMGHamiltonian}
\end{align}
where as before $S_x=\frac{1}{2}\sum_{k}\sigma_x(k)$ is the total spin operator in direction $x$ and so on. We are interested in $\lambda > 0$, standing for a ferromagnetic interaction, $\gamma\in (0, 1]$ describing the anisotropic in-plane coupling, and $B$ is the magnetic field along $z$ direction with $B \geq 0$. By a similar computation as before, we find
\begin{align}
H_{1/N}^{LMG}|_{\text{Sym}^N(\mathbb{C}^2)}=Q_{1/N}(h_0^{LMG}) + \frac{1}{N}Q_{1/N}(1) - \frac{3}{2N}Q_{1/N}(x^2+\gamma y^2),
\end{align}
where $h_0^{LMG}:=-\frac{1}{4}(x^2+\gamma y^2)-Bz$ denotes the principal symbol of $h_N:=h_0^{LMG}+N^{-1}(-\frac{3}{2}(x^2+\gamma y^2)+1)$.
\subsection{Classical limit}\label{claslim}
A continuous bundle of $C^*$- algebras provides a natural setting to describe models in quantum statistical mechanics \cite{Lan17}. By interpreting the semi-classical parameter $\hbar$ as the number of particles of a system, namely  $\hbar=1/N\in 1/\mathbb{N}^*\cup\{0\}$, the limit $N\to\infty$ provides the so-called {\em thermodynamic limit}, namely the density of the system $N/V$ is kept fixed, and the volume $V$ of the system sent to infinity, as well. The limiting system constructed at the limit $N=\infty$ is typically quantum statistical mechanics in infinite volume. In this setting the so-called
{\em quasi-local observables} are studied: these give rise to a non-commutative continuous bundles of $C^*$-algebras, denoted by $A^{(q)}$, which is defined over the base space $I= 1/\mathbb{N}^*\cup \{0\}$ with fibers at $1/N$ given by a $N$-fold tensor product of a matrix algebra with itself. However, the limit $N\to\infty$ can also provide the relation between classical (spin) theories viewed as limits of quantum statistical mechanics. In this case the quasi-symmetric (or macroscopic) observables are studied and these induce a commutative bundle of $C^*$- algebras $A^{(q)}$ which is defined over the same base space $I$ with exactly the same fibers at $1/N$ as the algebra $A^{(q)}$, but differ at $N=\infty$, i.e., at $1/N=0$. It is the bundle $A^{(c)}$ that relates these systems to strict deformation quantization, since macroscopic observables are precisely defined by (quasi-) symmetric sequences which form the continuous cross sections of a continuous bundle of $C^*$-algebras \cite{LMV,MV,Lan17}. In \cite[Theorem 2.3]{MV} it has been shown that the quantization maps \eqref{defquan3} indeed define quasi-symmetric sequences, and hence macroscopic observables. In what follows, therefore, we study the limit $N\to\infty$ with respect to these observables.

We then consider a sequence of unit eigenvectors $\psi_N\in\text{Sym}^N(\mathbb{C}^2)\subset\bigotimes_{n=1}^N\mathbb{C}^2$ of a mean-field quantum spin Hamiltonian $H_N|_{\text{Sym}^N(\mathbb{C}^2)}=Q_{1/N}(h_N)$ with $h_N$ some classical symbol as in \eqref{classicalsymbol1} for which each $h_k$ is a real polynomial $S^2\to\mathbb{R}$. These vectors induce a sequence of pure (permutation-invariant vector) states $(\omega_N)_N$, viz.  
$\omega_N(\cdot)=\langle\psi_N,(\cdot)\psi_N\rangle$. 
 In section \ref{classical limit of mean-field theories} we prove that under some conditions the sequence $(\omega_N)_N$ admits a {\bf classical limit}, in that
\begin{align}
\omega_0(f)=\lim_{N\to\infty}\omega_N(Q_{1/N}(f)), \label{classical limit}
\end{align}
exists for all polynomials $f\in\tilde{A}_0$ and and defines a state on $\tilde{A}_0$. 
\begin{remark}\label{extension}
We remark that the same limit (viz. equation \eqref{classical limit}) exists when we replace $\tilde{A}_0$ by $A_0$. To see this we stress that the quantization maps $Q_{1/N}$ defined by \eqref{defquan3} extend to all of $A_0$ by continuous linear extension, except for the Dirac-Groenewold-Rieffel condition. This follows from the fact that the base space $I$ is discrete outside zero and the maps $Q_{1/N}$ are uniformly bounded in $N$, in the sense that $||Q_{1/N}(f)||_N\leq ||f||_{\infty}$ for any $N\in\mathbb{N}$ \cite{Lan98}.  As a result the continuity properties in Remark \ref{continuitybundle} hold for all $f\in A_0$. 
\end{remark}
The state $\omega_0$ may be regarded as the classical limit of the family $\omega_N$. Using Riesz representation theorem \eqref{classical limit} means that
for all $f\in \tilde{A}_0$ one has
\begin{align}
 \mu_0(f)=\lim_{N\to\infty}\int_{S^2}d\mu_{\psi_N}(\Omega)f(\Omega),  \label{classicallimit3}
\end{align}
where $\mu_0$ is the probability measure corresponding to the state $\omega_0$ (i.e. $\omega_0(f)=\int_{S^2}d\mu_0f$) and $\mu_{\psi_N}$ is a probability measure on $S^2$ with density $B_{\psi_N}(\Omega):=|\langle v^{(\Omega)},\psi_N\rangle|^2$ called the {\bf Husimi density function} associated to the unit vector $\psi_N$. In other words 
$\mu_{\psi_N}$ is given by
\begin{align}
d\mu_{\psi_N}(\Omega)=d\Omega\frac{N+1}{4\pi}|\langle v^{(\Omega)},\psi_N\rangle|^2.
\end{align}
\begin{remark}\label{parameter}
We stress that the parameter $N$ at the same time denotes the number of spin $\frac{1}{2}$-particles described by the restricted mean-field Hamiltonian $H_N|_{\text{Sym}^N(\mathbb{C}^2)}$ as well as the total angular momentum $J=N/2$ of a single spin particle, already denoted in \cite{Mar}. Comparing this with the result by Lieb \cite{Lieb}, the 'thermodynamical' limit $N\to\infty$ and the 'classical' limit $J\to\infty$ are therefore taken at the same time. Note that this is based on the fact that the single-site algebra is fixed, i.e. we consider $M_2(\mathbb{C})$. If one instead considers a spin system on the $N$-fold tensor product of the the algebra $M_{2J+1}(\mathbb{C})$ with itself one can try to perform two limits $J,N\to\infty$, using the fact that  $M_{2J+1}(\mathbb{C})=Q_{1/J}(f)$ by surjectvity of \eqref{defquan3}. This goes beyond the scope of this paper.
\end{remark}
The plan of the paper is as follows. In section $2$ we state and prove our main results (Theorem \ref{convergence} with Corollary \ref{spectrumcor},  and Theorem \ref{THMMAIN}). Theorem \ref{convergence} and Corollary  \ref{spectrumcor} establish a connection between the spectrum of the quantum mean field spin in question sytems and the principal symbol on the classical Bloch sphere $S^2$. Theorem \ref{THMMAIN} instead shows the existence of the classical limit of a certain sequence of eigenvectors corresponding to a mean-field quantum spin system. In section \ref{SSB1} we give an application regarding spontaneous symmetry breaking, and in the appendix we provide a comprehensive overview of useful definitions and mathematical tools.
\section{Semiclassical properties of mean-field quantum spin systems}
\subsection{Asymptotic properties of the spectum}
In this subsection we assume that we are given a strict deformation quantization of a compact symplectic manifold in the sense of Berezin, in that the quantization map is defined in terms of a family of coherent states such that in particular properties \eqref{Berprop} and \eqref{identityunit} hold. Note that this clearly applies to $S^2$ where the quantization is defined by maps \eqref{defquan3}  as we have seen in the previous section. 
We now prove a result relating the spectrum of such quantization maps to the range of the function that is quantized.
\begin{theorem}\label{convergence}
Given a pure state Berezin strict deformation quantization of a compact symplectic manifold $X$ on base space $I$. Denote  by $Q_{1/N}$ the associated quantization maps. Then for any $f\in \tilde{A}_0=\text{Dom}(Q_{1/N})$ 
\begin{align}
\lim_{N\to\infty}\text{dist}\bigg{(}\text{ran}(f),\sigma(Q_{1/N}(f))\bigg{)}=0,
\end{align}
where $\sigma(Q_{1/N}(f))$ denotes the spectrum of the operator $Q_{1/N}(f)$, and $\text{dist}$ the distance function. In general the distance between a bounded set $X\subset\mathbb{C}$ and a nonempty set $Y\subset\mathbb{C}$ is
defined by
\begin{align}
\text{dist}(X,Y)=\sup_{x\in X}\inf_{y\in Y}|x-y|.
\end{align}
\end{theorem}
\begin{proof}
Let us assume by  contradiction that the statement in the theorem is not true. Then  there exists $\delta>0$, a function $f\in\tilde{A}_0$, and a sequence $(\lambda_{N_k})_k$ in $\text{ran}(f)$ such that $\text{dist}(\lambda_{N_k},\sigma(Q_{1/N_k}(f)))\geq\delta >0$ for all $k$. Since $X$ is compact and $f$ is continuous also $\text{ran}(f)$ is compact so that we can extract a subsequence $\lambda_{N_k'}$ converging to a point $r\in\text{ran}(f)$. 
Hence, for all $\epsilon>0$ there exists a $K_{\epsilon}$ such that $|r-\lambda_{N_k'}|<\epsilon$ for all $k'\geq K_{\epsilon}$. This implies that $r\notin\sigma(Q_{1/N_{k'}}(f))$ for $k'\geq K_{\epsilon}$, which means that the resolvent operator associated to $r$ and denoted by $R_{r}(Q_{1/N_{k'}}(f))$ exists for all $k'\geq K_{\epsilon}$.
\newline
\newline
Now, we can find an element $\Omega\in X$ such that $f(\Omega)=r$. 
By property \eqref{Berprop}, we can always recover $f(\Omega)$ as
\begin{align}
f(\Omega)=\lim_{N\to\infty}\langle v^{(\Omega)},Q_{1/N}(f) v^{(\Omega)}\rangle,\label{coherentstaterepresentation}
\end{align}
where $|v^{(\Omega)}\rangle$ denotes the coherent state induced by the point $\Omega\in X$ on which $f$ is defined. For this $\Omega$ and $k'\geq K_{\epsilon}$  let us now estimate
\begin{align}
&1=|\langle v^{(\Omega)},R_r(Q_{1/{N_k'}}(f))Q_{1/{N_k'}}(f-r)v^{(\Omega)}\rangle|^2\leq\nonumber\\&||R_r(Q_{1/{N_k'}}(f))||^2\cdot||Q_{1/{N_k'}}(f-r)v^{(\Omega)}||^2,\label{identity}
\end{align}
using in the first equality the fact that $Q_{1/{N_k'}}(f-r)=Q_{1/{N_k'}}(f)-Q_{1/{N_k'}}(1_X)r=Q_{1/{N_k'}}(f)-r$. We now make the following estimation on
\begin{align}
&||Q_{1/{N_k'}}(f-r)v^{(\Omega)}||^2=\langle v^{(\Omega)},(Q_{1/{N_k'}}(f-r))^*Q_{1/{N_k'}}(f-r)v^{(\Omega)}\rangle\nonumber=&\\&
\langle v^{(\Omega)},Q_{1/{N_k'}}(f-r)^*Q_{1/{N_k'}}(f-r)v^{(\Omega)}\rangle, \label{id4}
\end{align}
where we used that the quantization maps preserve the adjoint. Then,
\begin{align}
&\bigg{|}\langle v^{(\Omega)},Q_{1/{N_k'}}(f-r)^*Q_{1/{N_k'}}(f-r)v^{(\Omega)}\rangle-\langle v^{(\Omega)},Q_{1/{N_k'}}|f-r|^2v^{(\Omega)}\rangle\bigg{|}\nonumber\leq&\\&
\bigg{|}\bigg{|}Q_{1/{N_k'}}(f-r)^*Q_{1/{N_k'}}(f-r)- Q_{1/{N_k'}}(|f-r|^2)\bigg{|}\bigg{|},\label{Rieffellikecondition}
\end{align}
using the Cauchy-Schwarz inequality and the fact that $|v^{(\Omega)}\rangle$ are unit vectors. Since the quantization map is strict, we  conclude that the above inequality uniformly converges to zero as $k'\to\infty$ (see Remark \ref{continuitybundle}).
This together with \eqref{identity} and \eqref{id4} implies that
\begin{align*}
&1\leq||R_r(Q_{1/{N_k'}}(f))||^2\bigg{(} \langle v^{(\Omega)},Q_{1/{N_k'}}(f-r)^*Q_{1/{N_k'}}(f-r)|v^{(\Omega)}\rangle - \langle v^{(\Omega)},Q_{1/{N_k'}}(|f-r|^2) v^{(\Omega)}\rangle  \bigg{)} \\&+||R_r(Q_{1/{N_k'}}(f))||^2\langle v^{(\Omega)},Q_{1/{N_k'}}(|f-r|^2) v^{(\Omega)}\rangle\leq&\\&
||R_r(Q_{1/{N_k'}}(f))||^2\cdot\bigg{|}\bigg{|}Q_{1/{N_k'}}(f-r)^*Q_{1/{N_k'}}(f-r) - Q_{1/{N_k'}}(|f-r|^2)\bigg{|}\bigg{|}\\&+
||R_r(Q_{1/{N_k'}}(f))||^2\langle v^{(\Omega)},Q_{1/{N_k'}}(|f-r|^2) v^{(\Omega)}\rangle.
\end{align*}
By \eqref{coherentstaterepresentation} it follows that $\lim_{k'\to\infty}\langle v^{(\Omega)},Q_{1/{N_k'}}(|f-r|^2) v^{(\Omega)}\rangle=|f(\Omega)-r|^2=0$.  Since also $\lim_{k'\to\infty}||Q_{1/{N_k'}}(f-r)^*Q_{1/{N_k'}}(f-r)  - Q_{1/{N_k'}}(|f-r|^2)||=0$, it must follow that $||R_r(Q_{1/{N_k'}}(f))||^2\to\infty$ as $k'\to\infty$, which also implies that $||R_r(Q_{1/{N_k'}}(f))||\to\infty$ as $k'\to\infty$. In order to conclude we recall that that
\begin{align*}
||R_r(Q_{1/{N_k'}}(f))||_{N_k'}\leq\frac{1}{\text{dist}(r,\sigma(Q_{1/{N_k'}}(f))}.
\end{align*}
Combining the above inequalities yields for $k'$ large enough the final inequality
\begin{align}
0<\delta\leq\text{dist}(\lambda_{N_k'},\sigma(Q_{1/{N_k'}}(f)))\leq\frac{1}{||R_r(Q_{1/{N_k'}}(f))||_{N_k'}}.
\end{align}
By taking the limit $k'\to\infty$ of the above inequality we clearly arrive at  a contradiction, since the right-hand side converges to zero. 
This proves the theorem.
\end{proof}
By remark \ref{extension} this result extends to all of $A_0$. It yields an important corollary, relating the spectrum of a mean field quantum spin system with Hamiltonian $H_{1/N}$ to the range of a polynomial on $S^2$. As already mentioned in the introduction, the operator $H_{1/N}$ leaves $\text{Sym}^N(\mathbb{C}^2)$ invariant, and we therefore consider its restriction $H_{1/N}|_{\text{Sym}^N(\mathbb{C}^2)}$ which typically assumes the form $Q_{1/N}(h_N)$, where $Q_{1/N}$ is defined by \eqref{defquan3}, $h_N$ is a classical symbol of the type $h=h_0+\sum_{k=1}^MN^{-k}h_k$, with each $h_i \ (i=0,...,M)$ a real polynomial on $S^2$. With a slight abuse of notation, unless specified otherwise, let us now write $H_{1/N}:=H_{1/N}|_{\text{Sym}^N(\mathbb{C}^2)}$. Using this notation we prove the following result.
\begin{corollary}\label{spectrumcor}
Given a mean-field quantum spin system with Hamiltonian $H_{1/N}=Q_{1/N}(h_N)$, where $h_N$ is a classical symbol of the type $h=h_0+\sum_{k=1}^MN^{-k}h_k$, and each $h_i \ (i=0,...,M)$ is a real polynomial on $S^2$. Then, the spectrum of $H_{1/N}$ is related to the range of the principal symbol $h_0$ in the following way,
\begin{align}
\lim_{N\to\infty}\text{dist}(\text{ran}(h_0),\sigma(H_{1/N}))=0.
\end{align}
\end{corollary}
\begin{proof}
This follows for example from Weyl's Theorem (see e.g. \cite{Weyl} for details) applied to the hermitian matrices $Q_{1/N}(h_0)$ and $Q_{1/N}(\sum_{k=1}^MN^{-k}(h_k))=\sum_{k=1}^MN^{-k}Q_{1/N}(h_k)$, stating that if $\lambda_N^{(i)}$ is the $i^{th}$ eigenvalue of $H_{1/N}=Q_{1/N}(h_0)+\sum_{k=1}^MN^{-k}Q_{1/N}(h_k)$, and $\epsilon_N^{(i)}$ is the $i^{th}$ eigenvalue of $Q_{1/N}(h_0)$, then
\begin{align*}
|\lambda_N^{(i)}-\epsilon_N^{(i)}|\leq ||Q_{1/N}(\sum_{k=1}^MN^{-k}(h_k))||\leq\frac{1}{N}\max_{1\leq k\leq M}||h_k||_{\infty}\to 0 \ \ (N\to\infty),
\end{align*}
where in the last step we used that $||Q_{1/N}(h_k)||\leq ||h_k||_{\infty}$. In particular, we conclude that $\lim_{N\to\infty}\text{dist}(\sigma(Q_{1/N}(h_0)),\sigma(H_{1/N}))=0$. By the triangle inequality applied to the distance function in the previous theorem, the result follows.
\end{proof}

\subsection{Classical limit of mean-field theories}\label{classical limit of mean-field theories}
In this subsection we prove the existence of the classical limit. Since this paper focusses on mean-field quantum spin systems we state the results for the case where the manifold equals $S^2$. We remark that the main results can be generalized for arbitrary compact quantizable Kahler manifolds \cite{Porte}. We start with a lemma.
\begin{lemma}
Given a strict deformation quantization of $S^2$ with associated quantization maps $Q_{1/N}(h_0)$ defined by \eqref{defquan3}, where $h_0$ is a real polynomial function on $S^2$. Suppose that $E\in\text{ran}(h_0)$ such that $E=h_0(\Omega)$ for some finitely many distinct $\Omega\in S^2$, then for these $\Omega$ one has
\begin{align}
\lim_{N\to\infty}||Q_{1/N}(h_0)v^{(\Omega)}-Ev^{(\Omega)}||=0.
\end{align}
\end{lemma}
\begin{proof}
Consider those $\Omega\in S^2$ with $h_0(\Omega)=E$. Since there are finitely many of such $\Omega$, say $n$, we can label them $\Omega_1,...,\Omega_n$. By uniform continuity, given $\epsilon>0$, we can find $\delta>0$ such that for all $\Omega'\in B_{\delta}(\Omega_i)=\{\Omega'\in S^2 | \ d_{S^2}(\Omega',\Omega_i)<\delta \}$ one has $|h_0(\Omega')-h_0(\Omega_i)|<\epsilon/2$ for all $i=1,...,n$. Since the  $\Omega_i$ are distinct we can shrink each $B_{\delta}(\Omega_i)$ to make them pairwise disjoint, as is always possible in a Hausdorff space. For any such  $\Omega_i$ we compute
\begin{align*}
&\bigg{|}\bigg{|}Q_{1/N}(h_0)v^{(\Omega_i)}-E v^{(\Omega_i)}\bigg{|}\bigg{|} = \bigg{|}\bigg{|}\frac{N+1}{4\pi}\int_{S^2}d\Omega'(h_0(\Omega')-h_0(\Omega_i))\langle v^{(\Omega')},v^{(\Omega_i)}\rangle v^{(\Omega')} \bigg{|}\bigg{|}
\\&\leq\bigg{|}\bigg{|}\frac{N+1}{4\pi}\int_{B_{\delta}(\Omega_i)} d\Omega'(h_0(\Omega')-h_0(\Omega_i))\langle v^{(\Omega')},v^{(\Omega_i)}\rangle v^{(\Omega')}\bigg{|}\bigg{|} \\& + \bigg{|}\bigg{|}\frac{N+1}{4\pi}\int_{S^2\setminus B_{\delta}(\Omega_i)}d\Omega'(h_0(\Omega')-h_0(\Omega_i))\langle v^{(\Omega')},v^{(\Omega_i)}\rangle v^{(\Omega')}\bigg{|}\bigg{|}\\&\leq \sup_{\Omega'\in B_{\delta}(\Omega_i)}|h_0(\Omega')-h_0(\Omega_i)|\frac{N+1}{4\pi}\int_{B_{\delta}(\Omega_i)} d\Omega'|\langle v^{(\Omega')},v^{(\Omega_i)}\rangle| \\& + \frac{N+1}{4\pi}\int_{S^2\setminus B_{\delta}(\Omega_i)}d\Omega'|h_0(\Omega')-h_0(\Omega_i)||\langle v^{(\Omega')},v^{(\Omega_i)}\rangle|.
\end{align*}
The first term is bounded by $\epsilon/2$ using uniform continuity of $h_0$, and the fact that $1=\frac{N+1}{4\pi}\int_{S^2}d\Omega| v^{(\Omega)}\rangle\langle v^{(\Omega)}|$. For the second addendend, we note that the overlap between two coherent states is given by
\begin{align}
|\langle v^{(\Omega)},v^{(\Omega')}\rangle|=\bigg{(}\frac{1+t}{2}\bigg{)}^{N/2},
\end{align}
where $t\in [-1, 1)$ denotes the cosine of the angle between the both (different) coherent states. For $\Omega'\in S^2\setminus U_{\delta}$, we have $\frac{1+t}{2}\leq M_{\delta}<1$ for some constant $M_{\delta}$ depending on $\delta$. For $N$ sufficiently large, one can estimate
\begin{align}
\bigg{(}\frac{1+t}{2}\bigg{)}^{N/2}\leq M_{\delta}^{\frac{N}{2}}< \epsilon/2C.
\end{align}
where $C$ is the integration constant given by $C:=\frac{N+1}{4\pi}\int_{S^2\setminus B_{\delta}(\Omega_i)}d\Omega'|h_0(\Omega')-h_0(\Omega)|$. Hence, we conclude that for all $i=1,...,n$ one has
$\lim_{N\to\infty}||Q_{1/N}(h_0)v^{(\Omega_i)} - Ev^{(\Omega_i)}||=0$.
\end{proof}
\begin{remark}\label{spectrumeigenvalues}
In particular, if $E\in\text{ran}(h_0)$, then there is a sequence $\lambda_N$ of eigenvalues of $Q_{1/N}(h_0)$ such that $|E-\lambda_N|\to 0$ as $N\to\infty$. The same holds when we replace the principal symbol $h_0$ by $h_N$. Indeed, since such $h_N$ is typically given by $h=h_0+\sum_{k\geq 1}^MN^{-k}h_k$ for some $M\in\mathbb{N}$, by linearity $Q_{1/N}(h_N)=Q_{1/N}(h_0)+\sum_{k\geq 1}^MN^{-k}Q_{1/N}(h_k)$. Since each $||Q_{1/N}(h_k)||\leq ||h_k||_{\infty}$, clearly one has $||\sum_{k\geq 1}^MN^{-k}Q_{1/N}(h_k)||\leq\frac{1}{N}\max_{1\leq k\leq M}||h_k||_{\infty}$, and hence $||Q_{1/N}(h_N)v^{(\Omega)}-Ev^{(\Omega)}||\leq||Q_{1/N}(h_0)v^{(\Omega)}-Ev^{(\Omega)}||+O(1/N)$, so that by the lemma also $\lim_{N\to\infty}||Q_{1/N}(h_N)v^{(\Omega)}-Ev^{(\Omega)}||=0$.
\end{remark}
Before proving our main result regarding the existence of the classical limit of a mean field quantum spin system we assume the following data of the corresponding Hamiltonian. 
\begin{assumption}\label{symmetry}
Given a mean-field quantum spin system with Hamiltonian $H_{1/N}:=H_{1/N}|_{\text{Sym}^N(\mathbb{C}^2)}=Q_{1/N}(h_N)$ and polynomial symbol $h_N$ as defined before. We consider critical non-degenerate points $\Omega\in S^2$ of the principal symbol $h_0$ of $h_N$. By the Morse Lemma and the fact that $S^2$ is compact it follows that there are finitely many of such points $\Omega_i \ (i=1,...,d)$ which all are isolated, and hence distinct. We then fix an energy $E$ in the range of $h_0$ with the property that  $A(E):=h_0^{-1}(E)\cap \{\Omega_1,...,\Omega_d\}\neq \emptyset$, where $h_0^{-1}(E)$ is the preimage of $E$. Let us denote by $n$ the number of points in the set $A(E)$.
We assume that there exists a topological group $G$ of isometries of $S^2$ acting continuously on $S^2$ and $\Omega_1\in A(E)$ such that the corresponding orbit $\mathcal{O}(\Omega_1)$ equals $\{\Omega_1,...,\Omega_n\}$, where $\Omega_i \in A(E)$ for each $i=1,...,n$. This action obviously induces an action $T$ on the space of functions $\Psi$ on $C(S^2)$, namely
\begin{align}
(T_g\Psi)(\Omega):=\Psi(g^{-1}\Omega) \ \ (g\in G, \ \  \Omega\in S^2),
\end{align}
where the function $\Omega\mapsto \Psi(\Omega)$ is defined by $\Psi(\Omega):=\langle\Omega,\psi_N\rangle$ (see Appendix \ref{appA} for details). We furthermore assume that this action is such that for a particular sequence of normalized eigenvectors $(\psi_N)_N$ of $Q_{1/N}(h_N)$ one has 
\begin{align}
\Psi(g^{-1}\Omega)=\Psi(\Omega) \ \ (g\in G, \ \  \Omega\in S^2).
\end{align}
\end{assumption}
\begin{remark}
In Assumption \ref{symmetry} we considered non-degenerate critical points which correspond to some energy $E$. An interesting example is the case when $E$ is the minimum of $h_0$. If the Hessian of $h_0$ is nonsingular  at points in which $h_0$ attains the minimum value $E$ (i.e. the points are non-degenerate), one can try to prove that the sequence of ground state eigenvectors admits a classical limit, as we will see in the case of the Curie-Weiss model.  This is generally true for mean-field quantum spin systems as long as the Hesssian of the principal symbol of the model at all points of minima is nonsingular and the model meets the symmetry condition as explained above. If one instead takes a fixed energy $E$ with $\min(h_0)<E<\max(h_0)$ and assumes that $E$ is a regular energy (meaning that the gradient is non-zero at all points $\Omega$ in $S^2$ with $h_0(\Omega)=E$), then preimage is a curve with no singular points, and also now can try to prove the existence of the classical limit as in Theorem \ref{THMMAIN}.
This is a challenging problem and will not be further investigated in this paper. 
\end{remark}
\noindent
{\em Example 1 revisited.}\label{symmetrycw}\\
Let us consider the Curie-Weiss model of Example 1. We consider $E=\min(h_0)$. It is not difficult to see that $h_0^{-1}(E)=\{\Omega_-,\Omega_+\}$, where $\Omega_{\pm}$ are the two (distinct) minima of $h_0$. Indeed, the Hessian is nonsingular at  these points. The symmetry group is $\mathbb{Z}_2$ which consists of two elements, and the non-trivial element acts by reflexion. The orbit corresponding to any such minimal point generates both minima, and the (unique, up to phase) ground state eigenvector $\psi_N$ of $H_{1/N}^{CW}$ is $\mathbb{Z}_2$- invariant for each $N\in\mathbb{N}$ (see Corollary \ref{classicallimitcor} and \cite{VGRL18} for a detailed discussion). 
\\\\
Under Assumption \ref{symmetry} there exists a group $G$, $n$ distinct  points $\Omega_i\in S^2 \ (i=1,...,n)$ with $h_0(\Omega_i)=E$ such that $\mathcal{O}(\Omega_1)=\{\Omega_1,...,.\Omega_n\}$ (for some particular point $\Omega_1$), and we can find a sequence of $G$-invariant eigenstates $\psi\equiv(\psi_N)_N$ (with corresponding eigenvalues $\lambda_N$) of the operator $H_{1/N}$.  As a result of Remark \ref{spectrumeigenvalues} one clearly has $\lim_{N\to\infty}|\lambda_N - E|=0$. For these eigenstates we now prove the following result.
\begin{theorem}\label{THMMAIN}
Consider the sequence of algbraic states $(\omega_N)_N$ induced by the eigenstates $\psi\equiv(\psi_N)_N$ of the mean field operator $H_{1/N}=Q_{1/N}(h_N)$. Let $E\in\text{ran}(h_0)$ such that $|\lambda_N-E|\to 0$ as $N\to\infty$, where  $(\lambda_N)_N$ is the sequence of eigenvalues corresponding to $\psi_N$. Then,  the following limit exists and defines a state on $A_0$
\begin{align}
\omega_0(f)=\lim_{N\to\infty}\omega_N(Q_{1/N}(f)) \ \ (f\in {A}_0), \label{limit2}
\end{align}
where $\omega_0(f)=\frac{1}{n}\sum_{i=1}^nf(\Omega_i)$ and $\Omega_{i}$ are the points in $S^2$ such that $h_0(\Omega_i)=E$.
\end{theorem}
\begin{proof}
Clearly the symmetric combination $\frac{1}{n}\sum_{i=1}^nf(\Omega_i)$ where the $\Omega_{i}$ are such that $h_0(\Omega_i)=E$ in the above theorem defines a (mixed) state and hence a probability measure on $S^2$. By uniform continuity of $f$, given $\epsilon>0$, we can find $\delta>0$ such that for all $\Omega\in B_{\delta}(\Omega_i)=\{\Omega\in S^2 | \ d_{S^2}(\Omega,\Omega_i)<\delta \}$ one has $|f(\Omega)-f(\Omega_i)|<\epsilon/2n$, for all $i=1,...,n$. Let us fix one such neighborhood $U_1:=B_{\delta}(\Omega_1)$. By the assumption we can clearly find $g_i\in G$ ($i=1,...,n)$ so that the sets $U_i:=g_i^{-1}U_1$ are neighborhoods of $\Omega_i$, using the fact that  the map $\Omega\mapsto g\Omega$ is a homeomorphism. Since the points $\Omega_i$ are distinct the $n$ neighborhoods can be picked to be pairwise disjoint. 
It now follows that
\begin{align}
&\frac{N+1}{\pi}\int_{U_1}\frac{dz}{(1+|z|^2)^2}|\langle v^{(z)},\psi\rangle|^2=\frac{N+1}{4\pi}\int_{U_1}d\Omega|\Psi(\Omega)|^2=\nonumber\\&\frac{N+1}{4\pi}\int_{U_2}d\Omega|\Psi(\Omega)|^2=...=\frac{N+1}{4\pi}\int_{U_n}d\Omega|\Psi(\Omega)|^2\nonumber=\\&\frac{N+1}{\pi}\int_{U_n}\frac{dz}{(1+|z|^2)^2}|\langle v^{(z)},\psi\rangle|^2,\label{invariantnbhs}
\end{align}
where we used the property that for these $g_i\in G \ \ (i=1,...,n)$ and $\Omega\in U_1$, one has $\Psi(g_i^{-1}\Omega)=\Psi(\Omega)$ by assumption. 
This implies that $\Psi(U_1)=\Psi(g_i^{-1}U_1)=\Psi(U_i)$ for all $i=1,...n$, from which \eqref{invariantnbhs} follows.
Under stereographic projection the neighborhoods $U_i$ correspond to open sets in $\mathbb{C}$ which we will denoted by the same name, and the points  $S^2\ni \Omega_{i}$  will be denoted by $z_{i}\in\mathbb{C}$. Given now the sequence of normalized eigenstates $\psi\equiv(\psi_N)_N$ of $Q_{1/N}(h_N)$ inducing the pure states $(\omega_N)_N$, we compute $\omega_N(Q_{1/N}(f))$, i.e.
\begin{align}
&\langle\psi,Q_{1/N}(f)\psi\rangle_{\mathbb{C}^{N+1}}=\frac{N+1}{4\pi}\int_{S^2}d\Omega f(\Omega)|\langle\psi,v^{(\Omega)}\rangle|^2=\nonumber\\&\frac{N+1}{\pi}\int_{\mathbb{C}}\frac{dz^2}{(1+|z|^2)^{2}} f(z)|\Psi(z)|^2,
\end{align}
where as before $\Psi(z)=\langle v^{(z)}, \psi \rangle=(1+|z|^2)^{-N/2}p(z)$. We then split this integral as follows:
\begin{align}
\int_{\cup_{i}U_i}+ \int_{\mathbb{C}\setminus{\cup_i U_i}}\frac{N+1}{\pi}\frac{dz^2}{(1+|z|^2)^{2}} f(z)|\Psi(z)|^2.
\end{align}
Moreover, we use the coherent state property, i.e.
\begin{align}
1= \langle \psi,\psi\rangle = \frac{N+1}{\pi}\int_{\mathbb{C}}\frac{dz^2}{(1+|z|^2)^{2}} |\langle \psi,v^{(z)}\rangle|^2 .
\end{align}
We then write
\begin{align*}
&\lim_{N\to\infty}\bigg{|}\langle\psi,Q_{1/N}(f)\psi\rangle-\frac{1}{n}\sum_{i=1}^n f(\Omega_i)\bigg{|}\\=&\lim_{N\to\infty}
\bigg{|}\frac{N+1}{\pi}\int_{\mathbb{C}}\frac{dz^2}{(1+|z|^2)^{2}} f(z)|\langle v^{(z)},\psi\rangle|^2 - \frac{1}{n}\sum_{i=1}^n\frac{N+1}{\pi}\int_{\mathbb{C}}\frac{dz^2}{(1+|z|^2)^{2}} f(z_i)|\langle v^{(z)},\psi\rangle|^2\bigg{|}
\\=&
\lim_{N\to\infty}\bigg{|} \frac{N+1}{\pi}\int_{\mathbb{C}\setminus{\cup U_i}}\frac{dz^2}{(1+|z|^2)^{2}}\bigg{(}f(z)-\frac{1}{n}\sum_{i=1}^nf(z_i)\bigg{)}|\langle v^{(z)},\psi\rangle|^2\\&
+ \frac{N+1}{\pi}\int_{\cup U_i}\frac{dz^2}{(1+|z|^2)^{2}}\bigg{(}f(z)-\frac{1}{n}\sum_{i=1}^n f(z_i)\bigg{)}|\langle v^{(z)},\psi\rangle|^2
\bigg{|}
\end{align*}
We now estimate the first integral:
\begin{align}
&\bigg{|} \frac{N+1}{\pi}\int_{\mathbb{C}\setminus{\cup U_i}}\frac{dz^2}{(1+|z|^2)^{2}}\bigg{(}f(z)-\frac{1}{n}\sum_{i=1}^nf(z_i)\bigg{)}|\langle v^{(z)},\psi\rangle|^2\bigg{|}\nonumber\\&\leq \frac{(n+1)(N+1)}{\pi}||f||_{\infty}^2\int_{\mathbb{C}\setminus{\cup U_i}}\frac{dz^2}{(1+|z|^2)^{2}}|\Psi(z)|^2 \nonumber\\&=  (n+1)||f||_{\infty}^2||\psi||_{\mathbb{C}\setminus{\cup U_i}}^2.
\end{align} 
By the Theorem \ref{Portethm1} the norm $||\psi||_{\mathbb{C}\setminus{\cup U_i}}^2$  can be made bounded by $\epsilon/(2(n+1)||f||_{\infty}^2)$ as $N$ becomes sufficientely large, since outside the sets $U_{i}$ the distance to $\{h_0=E \}$ is clearly positive. This follows from the comments under Theorem \ref{Portethm1}, since for $N$ large enough we can estimate the norm $||\psi||_{\mathbb{C}\setminus{\cup U_i}}^2$ by $C/N$, for some constant $C$ that does not depend on $N$. 
We are thus done if we can show that 
\begin{align*}
&\bigg{|}\frac{N+1}{\pi}\int_{\cup U_i}\frac{dz^2}{(1+|z|^2)^{2}}f(z)|\langle v^{(z)},\psi\rangle|^2 -
\frac{1}{n}\sum_{i=1}^n\frac{N+1}{\pi}\int_{\cup U_i}\frac{dz^2}{(1+|z|^2)^{2}}f(z_i)|\langle v^{(z)},\psi\rangle|^2\bigg{|}< \epsilon/2,
\end{align*}
whenever $N$ sufficiently large.
Since $U_i\cap U_j=\emptyset$ if $i\neq j$, we can write
\begin{align*}
&\frac{N+1}{\pi}\int_{\cup U_i}\frac{dz^2}{(1+|z|^2)^{2}}f(z)|\langle v^{(z)},\psi\rangle|^2\\&=\sum_{i=1}^n\frac{N+1}{\pi}\int_{U_i}\frac{dz^2}{(1+|z|^2)^{2}}f(z)|\langle v^{(z)},\psi\rangle|^2.
\end{align*}
Making use of \eqref{invariantnbhs}  we can estimate
\begin{align}
&\bigg{|}\frac{N+1}{\pi}\int_{\cup U_i}\frac{dz^2}{(1+|z|^2)^{2}}f(z)|\langle v^{(z)},\psi\rangle|^2-\frac{1}{n}\sum_{i=1}^n
\frac{N+1}{\pi}\int_{\cup U_i}\frac{dz^2}{(1+|z|^2)^{2}}f(z_i)|\langle v^{(z)},\psi\rangle|^2\bigg{|}\nonumber
\\&\leq
\sum_{i=1}^n\bigg{|}\frac{N+1}{\pi}\int_{U_i}\frac{dz^2}{(1+|z|^2)^{2}}(f(z)-f(z_i))|\langle v^{(z)},\psi\rangle|^2\bigg{|}.\label{finalin}
\end{align}
By the fact that $\frac{N+1}{\pi}\int_{\mathbb{C}}\frac{dz^2}{(1+|z|^2)^{2}}|\Psi(z)|^2=1$,  we find that each of the $n$ terms is bounded by $\sup_{z\in U_i}|f(z)-f(z_i)|$. By uniform continuity of $f$, each supremum is bounded by  $\epsilon/2n$, so that \eqref{finalin} is bounded by $\epsilon/2$.  This shows that
\begin{align}
\bigg{|}\langle\psi,Q_{1/N}(f)\psi\rangle-\frac{1}{n}\sum_{i=1}^n f(\Omega_i)\bigg{|}< \epsilon,
\end{align}
whenever $N$ sufficiently large, which concludes the proof of the theorem. 
\end{proof}
\subsection{Spontaneous symmetry breaking}\label{SSB1}
Let us now say some words about spontaneous symmetry breaking (SSB). We refer to \cite{Lan17,LMV,VGRL18} for a more general discusision. For finite quantum systems the ground state of a generic Hamiltonian is typically unique and hence invariant under whatever symmetry group $G$ it may have. Hence SSB, in the sense of having a family of asymmetric ground states related by the action of $G$, seems possible only in infinite quantum systems or in classical systems (for both of which the arguments proving uniqueness). This is not what happens in nature, since real samples are finite and already do display SSB. In appendix \ref{appC} a possible mechanism is provided yielding some approximate form of symmetry breaking in large but finite systems, even though the theory seems to forbid this. We instead work with a definition of SSB that is standard in mathematical physics and applies equally to finite and infinite systems (provided these are described correctly), and to classical and quantum systems, namely that the ground state (as defined in Subsection \ref{claslim}) of a system with $G$-invariant dynamics (where $G$ is some group, typically a discrete group or a Lie group) is either pure but not $G$-invariant, or $G$-invariant but mixed. Accordingly, if SSB occurs this then implies that the exact pure ground state of a finite quantum system converges to a mixed state on the limit system. At least in the language of algebraic quantum theory this is the essence of SSB:
\\\\
{\em Pure ground states are not invariant, whilst invariant ground states are not pure.}

Under Assumption \ref{symmetry} and the assumptions of Theorem \ref{THMMAIN} we moreover assume that the sequence of eigenvectors $(\psi_N)_N$ corresponds to non-degenerate ground states (i.e. eigenvectors corresponding to the lowest eigenvalue such that the corresponding eigenspace is one-dimensional) of the mean field Hamiltonian $H_{1/N}|_{\text{Sym}^N(\mathbb{C}^2)}$. This just means that these eigenvectors $\psi_N\ (N\in\mathbb{N})$ are unique, up to a phase factor. One should notice that (ground state) eigenvectors of the compressed Hamiltonian $H_{1/N}|_{\text{Sym}^N(\mathbb{C}^2)}$ do a priori not correspond to (ground state) eigenvectors of the original Hamiltonian $H_{1/N}$ with the same eigenvalue. This only happens when the (ground state) eigenvector is permutation-invariant, i.e. an element of $\text{Sym}^N(\mathbb{C}^2)$ as is for example the case for the Curie-Weiss model. Then, clearly for each $N\in\mathbb{N}$ the corresponding algebraic pure state $\omega_N$ is $G$-invariant, implying there is no symmetry breaking for any finite $N$. However, the classical limiting state which also qualifies as a ground state in the algebraic sense \cite{vandeVen}, is invariant but not pure. In view of the above this explains the presence of SSB in the classical limit.
In the next corollary we give an illustrating example.
\begin{corollary}\label{classicallimitcor}
The Curie-weiss model $H_{1/N}^{CW}$ defined in \eqref{CWhamscaled} admits a classical limit which breaks the $\mathbb{Z}_2$-symmetry in the limit as $N\to\infty$.
\end{corollary}
\begin{proof}
In \cite{LMV,VGRL18} it has been shown that for any finite $N$ the ground state eigenvector $\psi_N$ of $H_{1/N}^{CW}$ is (up to phase) unique and lies in the symmetric subspace $\text{Sym}^N(\mathbb{C}^2)$. In view of Example $1$ in Subsection \ref{Mean field theories and symbol}, for any finite $N$ the vector $\psi_N$ is the ground state of $Q_{1/N}(h_N)=H_{1/N}^{CW}|_{\text{Sym}^N(\mathbb{C}^2)}$, where $h_N=h_0^{CW}+N^{-1}(-3Jz^2+1)$. By Example 1 and Theorem \ref{THMMAIN} this sequence admits a classical limit. In particular this limit is given by $\omega_0(f)=\frac{1}{2}(\omega_0^-(f)+\omega_0^+(f))$, where $\omega_0^{\pm}(f)=f(\Omega_{\pm})$ denotes the pure states localized at the two minima $\Omega_{\pm}:=(1/2,0,\pm\sqrt{3}/2)$ of the classical Curie-Weiss model $h_0^{CW}$. The state $\omega_0$ is a mixed state which also qualifies as a ground state, and it is obviously invariant under reflexion $\mathbb{Z}_2$- symmetry: $(x,y,z)\mapsto (x,-y,-z)$ (we refer to the remarks in \cite[Section 4.1]{LMV} for details). In summary, for any finite $N$ the algebraic ground state $\omega_{N}(\cdot)=\langle\psi_N,(\cdot)\psi_N\rangle$ is pure and $\mathbb{Z}_2$-invariant, whilst the limiting state $\omega_0$ is invariant but mixed. It may be clear that no pure invariant ground states of the classical Curie-Weiss model exist, since $\omega_0^{\pm}$ is mapped to $\omega_0^{\mp}$ under this symmetry. In view of the above dicussion the $\mathbb{Z}_2$-symmetry is therefore spontaneously broken in the limit $N\to\infty$. 
\end{proof}
\section*{Acknowledgments} The author is very grateful to Valter Moretti for the useful technical discussions, and also to Stéphane Nonnenmacher and Alix Deleporte and for their comments. The author also thanks the referees for their remarks. The author is a Marie Sk\l odowska-Curie fellow of the Istituto Nazionale di Alta Matematica and is funded by the INdAM Doctoral Programme in Mathematics and/or Applications co-funded by Marie Sk\l odowska-Curie Actions, INdAM-DP-COFUND-2015, grant number 713485. 
\appendix
\section{Appendices}
\subsection{Toeplitz quantization}\label{appA}
The following definition can also be found in \cite{Porte}. We shall briefly summarize the main parts.
\begin{definition}
\begin{enumerate} 
\item  A compact Kähler manifold $(M, J, \omega)$ is said to be quantizable when
the symplectic form $\omega$ has integer cohomology: there exists a unique Hermitian line bundle $(L, h)$ over $M$ such that the curvature of $h$ is $-2i\pi\omega$. This line bundle is called the prequantum line bundle over $(M, J, \omega)$. The manifold $(M, J, \omega)$ is said to be real-analytic when $\omega$
(or, equivalently, $h$) is real-analytic on the complex manifold $(M, J)$.
\item Let $(M, J, \omega)$ be a quantizable compact Kähler manifold with $(L, h)$
its prequantum bundle and let $N\in\mathbb{N}$.
\\
The Hilbert space $L^2(M,L^{\otimes N})$ is the closure of $C^{\infty}(M, L^{\otimes N} )$,
the space of smooth sections of the $N$-th tensor power of $L$. From the Hermitian metric $h$ on $L$, one deduces a Hermitian metric $h_N$ on $L^{\otimes N}$.  If $u,v$ are sections of $L^{\otimes N}$, the scalar product is defined as
\begin{align}
\langle u,v\rangle_{N}=\int_M\langle u, v\rangle_{h_N}\frac{\omega^{\wedge \text{dim}_{\mathbb{C}}M}}{(\text{dim}_{\mathbb{C}}M)!},
\end{align}
where $\langle u,v\rangle_{h_N}=h_N(u_1 \otimes\cdot\cdot\cdot\otimes u_N,  v_1 \otimes\cdot\cdot\cdot\otimes v_N)=\Pi_{m=1}^N h(u(m), v(m))$.
\item  The Hardy space  is the space of holomorphic sections of $L^{\otimes N}$, denoted by
denoted $H_0(M,L^{\otimes N})$. It is a finite-dimensional closed subspace of $L^2(M,L^{\otimes N})$. The Bergman projector $S_N$ is the orthogonal projector from the space of square-integrable sections $L^2(M,L^{\otimes N})$ to $H_0(M,L^{\otimes N})$.
If $f\in C^{\infty}(M,\mathbb{C})$, then the {\bf contravariant Toeplitz operator} $T_N(f) : H_0(M,L^{\otimes N})\to H_0(M,L^{\otimes N})$ is defined as
\begin{align}
T_N(f)u = S_N(fu).
\end{align}
\end{enumerate}
\end{definition}
\subsection{The case $M=S^2$}\label{appA1}
We consider the projective complex line $M=\mathbb{C}\mathbb{P}^1$ which is a Kähler manifold endowed with the Fubini-Study form $\omega_{FS}$, given by
\begin{align}
\omega_{FS}=\frac{dz\land d\overline{z}}{(1+|z|^2)^2},
\end{align}
with associated Kähler potentail $\phi$ defined by
\begin{align}
\phi(z,\overline{z})=\log{(1+|z|^2)}.
\end{align}
We can identify $\mathbb{C}\mathbb{P}^1\cong \mathbb{S}^2$, the Riemann sphere $\mathbb{S}^2$. Note that the Riemann sphere $\mathbb{S}^2$ is canonically embedded in $\mathbb{R}^3$ as the unit $2$-sphere: $S^2=\{(x,y,z) \ |\ x^2+y^2+z^2=1\}$. It is clear that $M$ is compact and quantizable since by a computation one finds that the integral of $\omega_{FS}$ on $\mathbb{C}\mathbb{P}^1$ is $2\pi$ times an integer. It turns out that the prequantum line bundle $L$ over $M=\mathbb{C}\mathbb{P}^1$ is the dual of the tautological bundle, that is, $L = O(1)$. In this case $L^{\otimes N} = O(N)$. In particular the space of holomorphic sections  of $L^{\otimes N}$ is identified with
the space of holomorphic square integrable functions on $\mathbb{C}$ with respect to the scalar product defined below. Indeed, unfolding the definitons for the Kahler manifold $M$, we find that the following scalar product indentifies a Hilbert space structure on $H_0(M,L^{\otimes N})$:
\begin{align}
\langle f, g\rangle=\int_{\mathbb{C}}\frac{f(z)g(z)}{(1+|z|^2)^{N+2}}d^2z,  
\end{align}
where $f,g$ are holomorphic functions on $\mathbb{C}$. We should remark that only holomorphic polynomials of degree $\leq N$ have a finite $L^2$-norm for this scalar product. We therefore only consider polynomials. An orthonormal basis consists of the monomials
\begin{align}
e_k=\sqrt{\frac{N+1}{\pi}}\sqrt{\binom{N}{k}}X^k, \ \ (k=0,...,N).
\end{align}
Analogously, one can define an inner product by absorbing the factor $\frac{N+1}{\pi}$, i.e., one considers the basis given by $w_k=\sqrt{\binom{N}{k}}X^k$ and scalar product
\begin{align}
\langle p_1, p_2\rangle_{\mathcal{F}\mathcal{B}^N}:=\frac{N+1}{\pi}\int_{\mathbb{C}}\frac{\overline{p_1(z)}p_2(z)}{(1+|z|^2)^{N+2}}d^2z. \label{scalar1}
\end{align}
This scalar product induces a $(N+1)$- dimensional Hilbert space, called the {\em Fock-Bargmann space $\mathcal{F}\mathcal{B}^N$} \cite{Com,Gaz}.
Consider now a vector $\psi\in\mathbb{C}^{N+1}$. To this vector we associate the function
\begin{align}
\Psi(z)=\langle v^{(z)},\psi\rangle,
\end{align}
where $|v^{(z)}\rangle$ i  a coherent state in complex coordinate $z$ (playing the same role as $|v^{(\Omega)}\rangle$). It can be shown that the function $\Psi(z)$ assumes the form 
\begin{align}
\Psi(z)=(1+|z|^2)^{-N/2}p(z), \label{formpsi}
\end{align}
where $p(z)$ is an analytic polynomial of dergee $\leq N$ in the Fock Bargmann space $\mathcal{F}\mathcal{B}^N$. Also the functions $\Psi(z)$ of the type \eqref{formpsi} form a finite-dimensional Hilbert space of dimension $N+1$, which are square integrable with respect to the scalar product
\begin{align}
\langle \Psi, \Phi\rangle_{\mathbb{C}^{N+1}}:=\frac{N+1}{\pi}\int_{\mathbb{C}}\frac{\overline{\Psi(z)}\Phi(z)}{(1+|z|^2)^{2}}d^2z.
\end{align}
This clearly implies that $||p||_{\mathcal{F}\mathcal{B}^N}=||\Psi||_{\mathbb{C}^{N+1}}$. In fact, the map $\mathbb{C}^{N+1}\ni\psi\mapsto\Psi(z)\in\mathcal{F}\mathcal{B}^N$ is an isometry. Therefore, for any vectors $\psi,\phi\in\mathbb{C}^{N+1}$, one has
\begin{align}
\langle\psi,\phi\rangle_{\mathbb{C}^{N+1}}&=\frac{N+1}{\pi}\int_{\mathbb{C}}\frac{dz^2}{(1+|z|^2)^{2}}\overline{\Psi(z)}\Phi(z)\\&=\frac{N+1}{4\pi}\int_{S^2}d\Omega\overline{\Psi(\Omega)}\Phi(\Omega),
\end{align}
where $\Psi(z)$ defined above, and $\Psi(\Omega)=\langle\Omega,\psi\rangle$ and $\Phi(\Omega)=\langle\Omega,\psi\rangle$. The measure $d\Omega$ is the usual measure on the $2$-sphere. We now state a result obtained by Deleporte \cite{Porte}.
\begin{theorem}[Deleporte]\label{Portethm1}
Let $f$ be a real-analytic, real-valued function on $S^2$ and $E\in\mathbb{R}$. Let
$(\psi_N)_N$ be a normalized sequence of $(\lambda_N)_N$-eigenstates of $T_N (f)$ with $\lambda_N\to E$ as $N\to\infty$. Then, for every open set $V$ at positive distance from $\{f = E\}$
there exist positive constants $c$, $C$ such that, for every $N\geq 1$, one has
\begin{align}
||\psi_N||_{L^2(V)}\leq C\exp{(-cN)},
\end{align}
where the norm is the one induced by the scalar product \eqref{scalar1}.
We say informally that, in the forbidden region $\{f \neq E\}$, the sequence
$(\psi_N )_N$ has an exponential decay rate.
\end{theorem}
Let us now take a function of the form $f_N=f_0+\sum_{k=1}^MN^{-k}f_k$, where $f_0$ and $f_k \ (k=1,...,M)$ are real polynomials on $S^2$, and let us consider the operator $T_N (f_N)$ with $(\psi_N)_N$ a normalized sequence of eigenstates with corresponding eigenvalues $(\lambda_N)_N$, such that  $\lambda_N\to E$ as $N\to\infty$.  Modifying the proof of the above theorem still results in convergence of the sequence of eigenvectors but with another decay rate, i.e. one obtains that for every open set $V$ at positive distance from $\{f_0 = E\}$, there holds $||\psi_N||_{L^2(V)}=O(1/N)$.
In order to see this, let us first go back to the case when $f=f_0$. Taking a glance at the proof of Theorem \ref{Portethm1} (applied to $T_N(f_0)$, with eigenstates $\phi_N$ and eigenvalues $\epsilon_N$) one first defines a function $a\in C^{\infty}(S^2,\mathbb{R}_+)$ such that $\text{supp}(a)\cap \{f_0=E\}=\emptyset$, and $a=1$ on $V$. Moreover one considers a neighborhood $W$ such that  $\text{supp}(a)\subset\subset W\subset\subset \{f_0 \neq E\}$. On $W$ the function $f_0-E$ is bounded away from zero and one can consider the analytic covariant symbol $g$ which is such that covariant Topelitz operator $T_N^{cov}(g)$ is the analytic inverse on this neighborhood of $T_N(f_0-\lambda_N)$. One can prove that this symbol is well-defined and that $T_N^{cov}(g)$ is indeed the mircolocal inverse of $T_N(f_0-\lambda_N)$. Then it has been shown that
\begin{align}\label{symbolineq}
T_N(a)T_N^{cov}(g)T_N(f_0-\lambda_N)=T_N(a)+O(e^{-cN}),
\end{align}
for some small $c>0$.
This implies that the sequence of eigenstates $\phi_N$ of $T_N(f_0)$ satisfies 
\begin{align}
0=T_N(a)\phi_N+O(e^{-cN}),
\end{align}
so that  $T_N(a)\phi_N=O(e^{-cN})$. In particular, 
\begin{align}\label{estimateina}
&||\phi_N||^2_{L^2(V)}=\frac{N+1}{4\pi}\int_V|\Phi_N(\Omega)|^2d\Omega\leq\nonumber\\&\frac{N+1}{4\pi}\int_{S^2}|\Phi_N(\Omega)|^2a(\Omega)d\Omega=\nonumber\\&\langle\phi_N,T_N(a)\phi_N\rangle=O(e^{-cN}).
\end{align}
In the general case, where $f_N=f_0+\sum_{k=1}^MN^{-k}f_k$, clearly $||T_N(f_0)\psi_N||= ||T_N(f_N)\psi_N - \sum_{k=1}^MN^{-k}T_N(h_k)\psi_N||\leq ||T_N(f_N)\psi_N|| + C_1/N$, for some $C_1>0$. By a similar argument as in the proof of Corollary \ref{spectrumcor} also $|\lambda_N-\epsilon_N|\leq C_2/N$, for some $C_2>0$. Then, on the one hand by \eqref{symbolineq} we obtain
\begin{align}
T_N(a)\psi_N=T_N(a)T_N^{cov}(g)T_N(f_0-\epsilon_N)\psi_N+O(e^{-cN})\psi_N.
\end{align}
On the other hand, 
\begin{align}
&||T_N(a)T_N^{cov}(g)T_N(f_0-\epsilon_N)\psi_N||\nonumber=\\&||T_N(a)T_N^{cov}(g)T_N(f_N-\sum_{k=1}^MN^{-k}f_k-\lambda_N+\lambda_N-\epsilon_N)\psi_N||\nonumber\leq\\&||T_N(a)T_N^{cov}(g)T_N(f_N-\lambda_N)\psi_N|| + \sum_{k=1}^MN^{-k}||T_N(a)T_N^{cov}(g)T_N(f_k)\psi_N||\nonumber+\\ & |\lambda_N-\epsilon_N| ||T_N(a)T_N^{cov}(g)\psi_N||\nonumber\leq\\&\frac{\tilde{C}_1}{N}+\frac{\tilde{C}_2}{N},
\end{align}
using that $T_N(a)T_N^{cov}(g)T_N(f_N-\lambda_N)\psi_N=0$, $\sum_{k=1}^MN^{-k}||T_N(a)T_N^{cov}(g)T_N(f_k)\psi_N||	\leq \tilde{C}_1/N$, and similarly that $ ||(\lambda_N-\epsilon_N)T_N(a)T_N^{cov}(g)\psi_N||\leq \tilde{C}_2/N$ where the constants $\tilde{C}_1,\tilde{C}_2$ do not depend on $N$. Combining the above equations yields
\begin{align}
T_N(a)\psi_N=O(1/N).
\end{align}
Repeating the same argument as in \eqref{estimateina} gives the estimate $||\psi_N||^2_{L^2(V)}=O(1/N).$
\\\\
Furthermore, it can be shown that for any smooth $f$ on the sphere $S^2$, one has $Q_{1/N}(f)=T_N(f)$, where  $T_N(f)$ is the Toeplitz operator with symbol $f$ and the quantization maps $Q_{1/N}$ are defined by \eqref{defquan3} (see e.g. \cite[Prop. 6.8]{SCHL} for a proof and general approach). This means that the above machinery is perfectly applicable to mean-field qauntum spin systems.
\subsection{Coherent spin states}\label{appB}
 If $|\!\uparrow\rangle, |\!\downarrow\rangle$ are the eigenvectors of $\sigma_3$ in $\mathbb{C}^2$, so that
$\sigma_3|\!\uparrow\rangle=|\!\uparrow\rangle$ and $\sigma_3|\!\downarrow\rangle=- |\!\downarrow\rangle$, and where  $\Omega \in {S}^2$, with  polar angles  
$\theta_\Omega \in (0,\pi)$, $\phi_\Omega \in (-\pi, \pi)$, we then define the unit vector
\begin{align}\label{om1}
|v^{(\Omega)}\rangle_1  = \cos \frac{\theta_\Omega}{2} |\!\uparrow\rangle + e^{i\phi_\Omega}\sin   \frac{\theta_\Omega}{2} |\!\downarrow\rangle.
\end{align}
 If $N \in \mathbb{N}$, the associated {\bf $N$-coherent spin state} $|v^{(\Omega)}\rangle_N\in \mathrm{Sym}^N(\mathbb{C}^2)$, equipped with the usual scalar product $\langle \cdot ,\cdot \rangle_N$ inherited from    
$\bigotimes_{n=1}^N\mathbb{C}^2$,  is defined as follows \cite{Pe72}:
\begin{align}\label{om2}
|v^{(\Omega)}\rangle_N =  \underbrace{|v^{(\Omega)}\rangle_1 \otimes \cdots \otimes |v^{(\Omega)} \rangle_1}_{N \: times}.
\end{align}
An important property relevant for our computations was established in \cite{LMV} 
\begin{align}
f(\Omega)=\lim_{N\to\infty}\frac{N+1}{4\pi}\int_{S^2}d\Omega' f(\Omega')|\langle v^{(\Omega)},v^{(\Omega')} \rangle_N|^2, \ \ (f\in C(S^2))\:.\label{Berezinprop}
\end{align}
\subsection{Spontaneous symmetry breaking}\label{appC}
We stress that at first sight spontaneous symmetry breaking seems to be a paradoxical phenomenon: in nature, finite quantum systems, such as crystals, evidently display it, yet in theory it seems forbidden in such systems, since ground states of a finite quantum system are invariant under some symmetry inherited from the properties of the Hamiltonian. As argued in \cite{Lan17} this means that some approximate and robust form of symmetry breaking should already occur in large but finite systems, despite the fact that uniqueness of the ground state seems to forbid this. SBB in a classical system or an infinite quantum system should therefore be foreshadowed in the quantum system whose classical limit it is, at least for tiny but positive values of Planck’s constant $\hbar$. In order accomplish this it must shown that for finite $N$ or $\hbar>0$ the system is not in its ground state, but in some other state having the property that as $N\to\infty$ or $\hbar\to 0$ it converges in a suitable sense to a symmetry-broken ground state of the limit system, which is either an infinite quantum system or a classical system. 
Since the symmetry of a state is preserved under the limits in question (provided these are taken correctly), this implies that the actual physical state at finite $N$ or $\hbar>0$ must already break the symmetry. A possible mechanism to achieve this (originating with Anderson \cite{And}) is based on forming symmetry-breaking linear combinations of low-lying states whose energy difference vanishes in the pertinent limit. Another idea, dated to Jona-Lasinio \cite{Jona} is based on a symmetry-breaking perturbation of the ground state which should be small to begin with, and should disappear in the pertinent limit. Both mechanisms share the property that the pure ground state (of the perturbed Hamiltonian) converges to a symmetry-breaking pure ground state on the limit system (be it a classical system or an infinite quantum system). This would explain the fact that in nature one of the pure symmetry-breaking states $\omega_i(f):=f(\Omega_i)$ is found, rather than the unphysical mixture $\omega_0(f)=\frac{1}{n}\sum_{i=1}^nf(\Omega_i)$ as originally predicted by the theory (viz. e.g. Theorem \ref{THMMAIN}). We refer to \cite{VGRL18} for a detailed discussion.
\section{Data availability statement}
The data that support the findings of this study are available from the corresponding author upon reasonable
request.

\end{document}